\documentclass[11pt]{article}

\usepackage[numbers]{natbib}
\usepackage{xspace}
\usepackage[usenames,dvipsnames,table]{xcolor}
\usepackage{graphicx}
\usepackage{amsmath,amssymb,amsthm}
\usepackage[margin=1in]{geometry}
\usepackage[algoruled,noend]{algorithm2e}
\usepackage{pstricks}
\usepackage{comment}
\usepackage{enumitem}
\usepackage{thmtools} 
\usepackage{thm-restate}
\usepackage[colorlinks=true,urlcolor=blue,linkcolor=RoyalBlue,citecolor=OliveGreen]{hyperref}

\newtheorem{assumption}{Assumption}

\newtheorem{lemma}{Lemma}
\newtheorem{definition}{Definition}
\newtheorem{observation}{Observation}
\newtheorem{proposition}{Proposition}

\newcommand{\owner}{\textsc{Owner}}
\newcommand{\miner}{\textsc{Miner}}
\newcommand{\predecessor}{\textsc{Pred}}
\newcommand{\sha}{\text{HASH}}

\newcommand{\reals}{\mathbb{R}}

\usepackage{tikz}
\usetikzlibrary{positioning,arrows.meta,arrows,decorations.markings}
\tikzset{main node/.style={circle,fill=blue!20,draw,minimum size=1cm,inner sep=0pt},}

\usepackage{amssymb}
\usepackage{pifont}

\newcommand{\jnote}[1]{{\color{magenta} {(\sf Jonah's Note:} {\sl{#1}}     {\sf )}}}
\newcommand{\mattnote}[1]{{\color{blue}\sl{#1}}}
\newcommand{\anote}[1]{{\color{red} {(\sf Alex's Note:} {\sl{#1}}     {\sf )}}}

\begin{document}

\title{Formal Barriers to {Longest-Chain} Proof-of-Stake Protocols}
\date{}
\author{Jonah Brown-Cohen\thanks{Computer Science Department, UC Berkeley. Email:\tt{jonahbc@eecs.berkeley.edu}} \and Arvind Narayanan\thanks{Computer Science Department, Princeton University. Email:\tt{arvindn@cs.princeton.edu}} \and Christos-Alexandros Psomas\thanks{Computer Science Department, Carnegie Mellon University. Email:\tt{cpsomas@cs.cmu.edu}} \and S. Matthew Weinberg\thanks{Computer Science Department, Princeton University. Email:\tt{smweinberg@princeton.edu}. Supported by NSF CCF-1717899. Work done in part while the author was a Research Fellow at the Simons Institute for the Theory of Computing.}}

\maketitle

\begin{abstract}
The security of most existing cryptocurrencies is based on a concept called \emph{Proof-of-Work}, in which users must solve a computationally hard cryptopuzzle to authorize transactions (``one unit of computation, one vote''). This leads to enormous expenditure on hardware and electricity in order to collect the rewards associated with transaction authorization. \emph{Proof-of-Stake} is an alternative concept that instead selects users to authorize transactions proportional to their wealth (``one coin, one vote''). Some aspects of the two paradigms are the same. For instance, obtaining voting power in Proof-of-Stake has a monetary cost just as in Proof-of-Work: a coin cannot be freely duplicated any more easily than a unit of computation. However some aspects are fundamentally different. In particular, \emph{exactly} because Proof-of-Stake is wasteless, there is no inherent resource cost to deviating (commonly referred to as the ``Nothing-at-Stake'' problem).

In contrast to prior work, we focus on incentive-driven deviations (\emph{any} participant will deviate if doing so yields higher revenue) instead of adversarial corruption (an adversary may take over a significant fraction of the network, but the remaining players follow the protocol). The main results of this paper are several formal barriers to designing incentive-compatible proof-of-stake cryptocurrencies (that don't apply to proof-of-work). 
\end{abstract}

\addtocounter{page}{-1}
\newpage
\section{Introduction}\label{sec:intro}
Since Nakamoto's white paper in 2008 \cite{nakamoto2008bitcoin}, Bitcoin and other cryptocurrencies have become ubiquitous, with hundreds of billions of USD worth of various cryptocurrencies currently in circulation.\footnote{\url{https://coinmarketcap.com} } While existing technology is already remarkable, it remains an active research area on many fronts. This paper focuses on one aspect of this agenda: Proof-of-Work versus Proof-of Stake. We provide a brief overview of the salient aspects of cryptocurrencies below before highlighting our contributions.

\subsection{What is a Cryptocurrency?}
At their core, all cryptocurrencies are simply decentralized ledgers. A network of participants wish to agree upon a sequence of events and the order in which they occurred. These events could be monetary transactions in the case of Bitcoin, script commands in the case of Ethereum, or many others. Two salient features of cryptocurrencies that define their purpose are the following:
\begin{itemize}[leftmargin=*]
\item It is crucial that the entire network reach \emph{consensus} on the occurrence of events. For instance, a currency isn't very useful if users can't agree on who owns which coins. A shared virtual machine isn't very useful if users can't agree on its state. 
\item The network is \emph{permissionless and pseudonymous}. That is, no identification outside the network is necessary to join and participate.\footnote{``Permissioned blockchains'' are becoming popular in finance, but these are fundamentally different than permissionless cryptocurrencies.}
\end{itemize}

Basic cryptography ensures that users can't forge transactions from other accounts, or undetectably propose otherwise invalid commands, so the main challenge is ensuring that users all agree on a state of the ledger, and also that no adversary can unduly influence the ledger's state, or otherwise gain by subverting the protocol. The typical attack to keep in mind is called a ``double-spend:'' imagine that you wish to purchase a car with Bitcoin. You digitally sign a transaction paying the owner a large sum of Bitcoin and broadcast it. The network agrees on a state of history where this transaction occurs, and then you get the car keys. As soon as this happens, you announce a digitally signed transaction that pays the same Bitcoin to an alternate account that you control, and do your best to subvert the protocol to agree on a ledger that includes this new transaction instead, leaving you with your Bitcoin intact as well as your new car.\footnote{Note that no ledger can possibly contain both transactions, as they spend the same coins. A currency would have no use if it were possible to spend the same coins twice.} The point here is that the ledger \emph{is} the currency: if the network believes that a transaction didn't take place, then by definition that transaction didn't take place.

A strawman proposal to cope might sound like this: every ten minutes, a uniform random participant from the network is selected. They may output a list of any number of consistent, valid transactions they like, and broadcast this to the entire network, along with a hash pointer to a previous selectee's output. Each output is assigned some numerical score,\footnote{For instance, Bitcoin's ``longest-chain rule'' (essentially) assigns a score to output $B$ equal to $B$'s distance from the root (when following hash pointers). See Definition~\ref{def:longestchain} in Section~\ref{sec:prelim}.} and users are asked to ``believe'' the history associated with the highest-scoring output they've seen so far. The key problem with this proposal is that selecting a uniformly random participant in a permissionless environment is absurd: users can freely create many ``Sybils,'' making a uniformly random participant simply the participant who successfully created the most IDs. One of the key ingredients in successful cryptocurrencies is a random selection process that is \emph{Sybil-proof}. 

\subsubsection*{Proof-of-Work Versus Proof-of-Stake}
The two most popular approaches to the above challenge are termed ``Proof-of-Work'' and ``Proof-of-Stake.'' Proof-of-Work is employed by Bitcoin and Ethereum, and aims to select a user randomly, but \emph{proportional to their computational power}. The idea is that while an attacker can certainly go out and purchase more computational power, exaggerating your computational power comes at a cost (unlike creating additional IDs). Proof-of-Work is typically implemented by requiring that all messages are concatenated with a nonce such that $\sha(\text{message, nonce}) <<2^{256}$, for some ideal hash function $\sha$ with 256-bit output. It is widely believed that the best way to find such a nonce is to randomly guess (referred to as \emph{mining}), and so every unit of computational power gives you a small additional probability of being able to send a valid message. As a result, some estimates predict over 2 billion USD spent annually (between electricity, cooling, etc.) \emph{just computing hashes} (more than a fourth of the NSF budget for 2017)~\cite{narayanan2016bitcoin}.

Proof-of-Stake isn't as widespread, but is an ongoing research focus of Ethereum~\cite{ButerinG17}, and is still responsible for billions of USD through cryptocurrencies such as NEM~\cite{NEM}, Cardano,\footnote{https://www.cardano.org/en/home/} BlackCoin~\cite{vasin2014blackcoin}, PeerCoin~\cite{Peercoin}, Nxt~\cite{Nxt}, and Tezos~\cite{Tezos}. The goal is to select a user randomly, and \emph{proportional to their wealth} (in the currency itself). The idea again is that while an attacker can certainly go out and purchase more stake in the currency, it comes at a cost. In comparison to Proof-of-Work, Proof-of-Stake wastes no electricity. 

\subsection{Security Concerns Specific to Proof-of-Stake}
From a security perspective, it's actually convenient that Proof-of-Work wastes resources: this guarantees that certain deviations (discussed in the next paragraph) from the intended protocol also cost additional resources, and are naturally disincentivized (this, of course, does not mean that deviations are never profitable, see e.g.~\cite{eyal2014majority,eyal2015miner, sapirshtein2016optimal, carlsten2016instability}). Proof-of-Stake, on the other hand, has the property that such deviations consume no resources and this is \emph{exactly because Proof-of-Stake is wasteless}, so it falls on the protocol to disincentivize such behavior through clever reward schemes.

This phenomenon is commonly referred to as ``Nothing-at-Stake,'' and refers to the fact that it consumes no additional resources for participants to, for instance, copy an outdated history of the currency and participate simultaneously with the ``real'' one, or even to copy every outdated history and participate in all of them simultaneously. Numerous recent works, both commercial and academic, aim to address this challenge with clever reward schemes. Commercial protocols indeed propose interesting approaches, and don't seem to have suffered major security setbacks to date. However, these ideas appear informally in whitepapers, often without formal definitions or rigorous reasoning, and are far from being fully explored. Academic proposals, on the other hand, provide rigorous, provable security guarantees, usually in ``network intrusion'' models (where a fraction of the network is malicious, but the remainder is honest).

In contrast to both existing streams of literature, we focus on rigorous guarantees in a ``strategic'' model (100\% of users act to maximize their own reward). \textbf{The main result of this paper is a formal barrier to incentive-compatible proof-of-stake protocols}. Specifically, we introduce a model for Proof-of-Stake protocols that captures the vast majority of commercial and academic proposals. Next, we show that every protocol fitting into this model must satisfy one of two complementary properties. Finally, we show how an attacker can exploit each one of these properties in order to benefit by strategically deviating from the prescribed protocol.

At a conceptual level, the barriers stem from the following: all cryptocurrencies require some source of (pseudo)randomness. In Proof-of-Work, this pseudorandomness is in some sense external to the cryptocurrency: the first miner to successfully find a good nonce produces the next block, and this miner is selected completely independently of the current state of the cryptocurrency. 

In Proof-of-Stake, it is highly desirable that the pseudorandomness comes from within the cryptocurrency itself, versus an external source (due to network security concerns discussed in Section~\ref{sec:prelim}). One might initially suspect that with sufficiently many hashes or digital signatures of past blocks, this can indeed serve as a good source of pseudorandomness for future blocks. However, we formalize surprising barriers showing a fundamental difference between external pseudorandomness and pseudorandomness coming from the cryptocurrency itself.

\subsection{Our Model: Formal Guarantees in Ideal Network Conditions}
Seminal work of Fischer, Lynch, and Paterson~\cite{FischerLP85} proves that if one desires formal consensus guarantees, even in the presence of a single adversarial user, one must make \emph{some} network assumptions. Existing academic works therefore aim to make the minimal assumptions necessary.


Consider instead an ideal network where every node has perfectly synchronized clocks, and every message is received with zero latency by every other node. It is not hard to design secure Proof-of-Stake cryptocurrencies in this model (at least, ones that are secure against known attacks), but these protocols will look completely unrealistic. The key power that this assumption buys is the ability to ignore messages that arrive even slightly late, as they cannot have been sent by an honest participant.
The problem with \emph{targeting} a solution that is only secure in this ideal model is that you wind up with an extremely fragile protocol which can't handle even minimal latency. 

Still, despite how unrealistically strong this model is, the authors aren't aware of existing protocols without incentive issues \emph{even in this ideal network model} (put another way, it is surprisingly daunting to develop a protocol that is secure in the typical sense for any non-trivial network model, yet secure in the strategic sense in the ideal network model). Rather than posing a long list of similar looking attacks, we formalize intuitively undesirable properties of a Proof-of-Stake protocol that capture the issues in many existing proposals. For each property, we show that \emph{any Proof-of-Stake protocol with this property} is vulnerable to a certain kind of attack. We elaborate further on these properties and attacks in the technical sections.

\subsection{Comparison to Related Work}\label{sec:related}
The ``incentive-driven'' threat to cryptocurrencies appears well-understood even as far back as Nakamoto's whitepaper, yet has received considerably less formal attention (some notable exceptions unrelated to Proof-of-Stake include~\cite{babaioff2012bitcoin, eyal2015miner, eyal2014majority, kiayias2016blockchain, carlsten2016instability, sapirshtein2016optimal}). While such ``attacks'' aren't a direct threat to consensus, they pose a severe indirect threat:~\cite{eyal2014majority} observes that such attackers gobbling up profits from honest participants could drive them out of the market, enabling a threat to consensus.

The most obviously related works to the present paper are academic Proof-of-Stake proposals~\cite{bentov2016snow, kiayias2017ouroboros, gilad2017algorand}. These works focus primarily on the network intrusion threat model, as it is obviously important that proposals be secure in the classical sense before concerning oneself with incentives. 

Some works go further and provide incentive guarantees, proving that miners who strategically deviate from the prescribed protocol can only gain a small $\varepsilon$ fraction of the total rewards~\cite{bentov2016snow, kiayias2017ouroboros}. Still,~\cite{bentov2016snow} notes that it is preferable for known strategic deviations to be \emph{strictly} disincentivized (and prove that their scheme achieves this for Nothing-at-Stake), and~\cite{kiayias2017ouroboros} observes that not all known deviations are captured by such claims (and prove that their scheme successfully disincentivizes double-spending). In the context of these works, we propose that the deviations formalized in this paper receive similar treatment to currently-known attacks in future analyses. We provide much more detail regarding how our results interact with existing proposals in Appendix~\ref{sec:examples}.
\subsection{Summary of Contributions and Roadmap}
Our model aims to isolate the incentive-driven threat, and we formalize two complementary properties such that every longest-chain Proof-of-Stake protocol must satisfy one of them. We further demonstrate incentive-driven attacks against protocols satisfying either of these properties. 

Section~\ref{sec:prelim} contains necessary definitions related to cryptocurrencies and Proof-of-Stake (where even our formal definition of Proof-of-Stake may be of interest). Section~\ref{sec:properties} defines generic properties shared by a wide class of protocols. Section~\ref{SEC:ATTACKS} poses generic attacks against any protocol with the properties defined in Section~\ref{sec:properties}. Section~\ref{sec:discussion} provides more detailed context for our paper with respect to prior work and proposes future directions. 
\section{Preliminaries}\label{sec:prelim}

\begin{definition}[Block]
A cryptocurrency stores its decentralized ledger in a set of objects called \emph{blocks}. Every block $B$ contains a pointer to its predecessor $\predecessor(B)$, a previous block (via its \sha). Every block is created by a single \emph{miner}, denoted by $\miner(B)$, and has a timestamp $t_B$ that indicates its claimed creation time (that is, the creator of the block can insert any value they like for $t_B$, independent of the actual time at which it was created). Blocks also contain some other information that has semantic meaning (such as transactions in the case of Bitcoin, scripts in the case of Ethereum, etc.). 
\end{definition}

Each block $B$ describes a potential history of events, as defined by the semantic contents of $B$ and its predecessors. For example, a Bitcoin block describes a series of monetary transactions.

\begin{definition}[Coin]
The basic monetary unit of any cryptocurrency is called a \emph{coin}, referenced by a unique ID. Every coin has an owner, referenced by a unique public key. Certain transactions have an associated semantic meaning that changes the owner of a coin. So for a given block $B$, and coin $c$, one can define $\owner_B(c)$ to be the owner of a coin $c$ as defined by the semantic meaning of transactions included in the history defined by block $B$. Here, we are referring to the smallest discrete monetary unit (so one Satoshi in the case of Bitcoin, rather than one bitcoin). 
\end{definition}

\vspace{-2mm}

The key problem that all cryptocurrencies need to resolve is the \emph{consensus problem}: how can we get all nodes in the network to eventually agree on the occurrence of events? All such protocols necessarily have a notion of \emph{validity}. That is, users cannot just generate arbitrary messages and send them at arbitrary times and have the rest of the network recognize these as potential blocks to be added to history (for instance, the block could contain invalid transactions, or it might not be that user's ``turn'' to send a message). We focus on protocols that satisfy the following assumptions:


\begin{assumption}
\label{assumption:protocol}
Let $T(B)$ be the graph whose nodes are blocks who share an ancestor with $B$, and whose edges are pointers to predecessors. We assume that:
\begin{enumerate}[label=(\alph*),leftmargin=*]
\item (Chain Dependence) Block $B$'s validity at time $t$ only depends on $B$, $t$ and $B$'s predecessors.
\item (Monotonicity) If a block $B$ is valid at time $t$ for a given graph $T(B)$, then it is valid for all graphs $T'$ that contain $T(B)$ as a subgraph and for all times $t' \geq t$.
\end{enumerate}
\end{assumption}


Before continuing, let's quickly motivate/discuss Assumption~\ref{assumption:protocol}. All protocols that the authors are aware of satisfy Assumption~\ref{assumption:protocol}, and this is for good reason: all protocols are vulnerable to what are called ``Eclipse'' attacks. An Eclipse attack occurs when an adversary prevents or blocks messages to honest participants of the protocol. If an adversary can temporarily partition the network into disjoint sets, of course the network can't reach consensus while partitioned. However, one might hope that the network can reach consensus once reunited. A major barrier to this possibility would be if a user once believed block $B$ to be valid (and built a deep history on top of it), and only learned once reunited that in fact $B$ was invalid all along. Without Chain Dependence and Monotonicity, this situation is entirely possible. With Chain Dependence and Monotonicity, Eclipse attacks are still a threat, but at least they cannot trick a user into believing a block is valid only to discover later that it was in fact invalid. Toy examples are provided in Appendix~\ref{app:assumption} to aid the interested reader in parsing Assumption~\ref{assumption:protocol}.

Now, we begin restricting attention to Proof-of-Stake protocols, in which the owner of a coin $c$ is eligible to mine a new block $B$ at some time $t$ according to the rules of the protocol. That is, every block $B$ further references a coin, $c_B$ that is used to witness $B$'s validity.
For any protocol satisfying Assumption~\ref{assumption:protocol}, the validity of a block can be determined by a function that takes only a block $B$ and the current time $t$ as input. This is because all of $B$'s predecessors can be accessed by following back the predecessor pointers to the root (and $t_B, c_B$ are both included in $B$). Thus, under these assumptions, a Proof-of-Stake Protocol can be defined as follows:

\begin{definition}[Proof-of-Stake Protocol]\label{def:PoS}
A \emph{Proof-of-Stake Protocol} $P$ is fully defined via two deterministic functions: a validating function $V_P$ and a mining function $M_P$. The validating function satisfies the following requirements:
\begin{itemize}
\item $V_P$ takes as input a block $B$ (which includes the claimed time of creation, $t_B$, and the claimed coin witness, $c_B$), and outputs an element of $\{ 0 , 1 \}$. \item $V_P$ must be efficiently computable by every participant in the protocol. 
\item A block $B$ is \emph{valid} at time $t$ if and only if $\predecessor(B)$ is valid and
\[
V_P(B) \cdot \mathbb{I}\{ \owner_{\predecessor(B)}(c_B) = \miner(B) \} \cdot \mathbb{I} \{ t_B \in [t_{\predecessor(B)}, t] \} = 1.
\]
\end{itemize}
The mining function $M_P$ satisfies the following requirements:
\begin{itemize}
\item $M_P$ takes as input a block $A$, a coin $c$ and timestamp $t$ and outputs a block. 
\item $M_P(A,c,t)$ is efficiently computable by $\owner_A(c)$.
\item For any coin $c$ and any time $t$, if there exists a block $B$ such that $B$ is valid at time $t$, where $\predecessor(B) = A$, $c_B = c$, and $t_B = t$, then $M_P(A,c,t) = B'$ where $B'$ is valid, $t_{B'} = t$, $c_{B'}= c$, and $\predecessor(B') = A$.
\item For any coin $c$ and any time $t$, if there is no block $B$ such that $B$ is valid at time $t$, where $\predecessor(B) = A$, $c_B = c$, and $t_B  = t$, then $M_P(A,c,t) = \bot$.
\end{itemize}
\end{definition}

Again before continuing, let's parse some aspects of this definition and what separates Proof-of-Stake from Proof-of-Work. The first two $V_P$ bullets are uncontroversial: a block is valid or it isn't, and every user in the network better be able to tell which blocks are valid. The third bullet might at first appear confusing, but recall that $B$ contains a reference to $t_B$, the claimed time of creation, and $c_B$, the coin witnessing validity. So $V_P(\cdot)$ can in fact depend on these. Beyond that, the first indicator is necessary to ensure that miners can't cheat by moving the same coin around different public keys in a potential block in order to make that block itself valid.\footnote{Appendix~\ref{app:freeze} contains a brief discussion on ``freezing'' coins for longer than just a single block, which only requires modifying $\predecessor(B)$ to $\predecessor^T(B)$ for some $T > 1$.} The second indicator is necessary to guarantee that a miner can't claim to have produced a block before s/he heard about its predecessor, nor claim to have produced a block in the future (but otherwise the \emph{current} timestamp is irrelevant for determining a block's validity, due to Assumption~\ref{assumption:protocol}). 

As for $M_P$, the second bullet captures two salient features. The first is that $M_P(A, c, t)$ is efficiently computable. This is what separates Proof-of-Stake from Proof-of-Work: if $M_P$ were not efficiently computable (e.g. because it involved inverting an ideal hash function), it would require non-trivial work to mine (Proof-of-Work). Because $M_P$ is efficiently computable, the owner of coin $c$ need only run $M_P(A, c, t)$ once during timestep $t$ for each coin they own, and has nothing to gain by doing additional work. The second salient feature is actually an omission: that $M_P(A, c, t)$ is not necessarily efficiently computable by miners other than $\owner_A(c)$ (e.g. because it perhaps requires producing a digital signature). The other three bullets are straight-forward. 

So far we have only discussed the validity of blocks. We also need to discuss where an honest user should mine (i.e. which block $A$ should be input to $M_P$). The dominating paradigm among existing proposals (including all commercial protcols referenced in Section~\ref{sec:intro}, but excluding Algorand~\cite{micali2016algorand,gilad2017algorand}, Casper~\cite{ButerinG17}, and GHOST~\cite{sompolinsky2015secure}) are variants of the \emph{longest-chain protocol}, where each block is given a monotone increasing score, and nodes are asked to ``believe'' the highest scoring block.\footnote{GHOST~\cite{sompolinsky2015secure} technically cannot be phrased in this language because the score of a block depends on the existence of descendants of $B$'s ancestors not referenced directly in $B$. If instead blocks are required to include pointers to these other blocks in order to ``get credit,'' then GHOST would also fit in this language. See also Appendix~\ref{app:GHOST} for a further discussion of different concerns regarding GHOST and Proof-of-Stake.}

\begin{definition}[Longest-Chain Variant]\label{def:longestchain}
A \emph{Longest-Chain Variant} has an associated scoring function $S(\cdot)$. $S$ takes as input a block $B$, outputs a score $S(B)\in \mathbb{R}$, and is monotone increasing: if $B'=\predecessor(B)$, then $S(B) > S(B')$. A Longest-Chain-Variant consensus protocol associated with $S$ asks users to mine on the valid block $A$ maximizing $S(A)$ (among blocks they are aware of).\footnote{Ties are allowed to be broken arbitrarily, but consistently (i.e. if $S(A) = S(A')$, users can arbitrarily decide to adopt $A$ or $A'$, but cannot switch between adopting $A$ and $A'$).} 
\end{definition}


Longest-Chain Variants are particularly common within cryptocurrencies because of their robustness to Eclipse attacks. Even if the network is partitioned for an extended period, and both disjoint subsets produce completely different histories, the entire network will quickly converge to the higher-scoring history as soon as the subsets reunite. Alternative protocols based on Byzantine Consensus~\cite{micali2016algorand,ButerinG17} lack this property, and instead achieve \emph{finality}. That is, once a user considers a block $B$ to be included in the ledger, they will never consider valid any ledger that not including $B$. Indeed, in protocols with finality, if the network is partitioned for an extended period, progress will either stall, or the network will never reach consensus even after being reunited. 

\section{Properties of Protocols}\label{sec:properties}
In this section we introduce two simple, intuitively desirable properties for Proof-of-Stake protocols. 

\subsection{(Un)-Predictability}
The first desirable property we define is unpredictability. Intuitively, it is good for protocols to be unpredictable in the sense that miners do not learn that they are eligible to mine a block until shortly before it is due to be mined. Many attacks, such as double-spending (discussed in Section~\ref{sec:intro}), or selfish-mining (\cite{eyal2014majority}, discussed in Section~\ref{SEC:ATTACKS}), can become \emph{much} more profitable if miners know in advance when they become eligible to mine. We begin with the definition of local predictability, which describes protocols where the owner of a coin knows in advance if she is eligible to mine a block with that coin. Let $\predecessor^D(B)$ be the $D$-th predecessor of a block $B$. That is, $\predecessor^1(B) = \predecessor(B)$ is the block that $B$ is mined on top of, $\predecessor^2(B) = \predecessor\left( \predecessor(B) \right)$, etc.

\begin{definition}[$D$-locally predictable]
A coin $c$ is $D$-locally predictable at block $A$ for timestamp $t$ if $\owner(c)$ can efficiently predict whether or not there will exist a block $B$ with $c_B = c$ such that $V_P(B)=1$, where $\predecessor^D(B) = A$ and $t_B  = t$.
\end{definition}

\begin{observation}
\label{obs:1locallypredictable}
For any Proof-of-Stake protocol, every coin $c$ is $1$-locally predictable at every block $A$ for every timestamp $t > t_A$.
\end{observation}
\begin{proof}
Fix a coin $c$, block $A$, 
and timestamp $t > t_A$. If there is some block $B$ with $c_B = c$ such that $V_P(B)=1$, $t_B  = t$ and $\predecessor^D(B) = A$, then $M_P(A,c,t)$ outputs such a block. If not, then $M_P(A,c,t) = \bot$. Since $M_P$ is efficiently computable by $\owner(c)$, we have that the coin $c$ is 1-locally predictable at block $A$ for timestamp $t$.
\end{proof}




\noindent In many existing protocols, every coin in a given protocol will be $D$-locally predictable at every block and for every timestamp. In such cases, we will refer to the protocol itself as being D-locally predictable. Intuitively, local predictability captures that a miner can predict in advance when they will be able to produce a block (whereas in Proof-of-Work protocols, they learn only the instant that the block is produced). Global predictability is a stronger definition which describes protocols in which every participant knows in advance if the owner of a given coin is eligible to mine a block.

\begin{definition}[$D$-globally predictable]
A coin $c$ is $D$-globally predictable at block $A$ for timestamp $t$ if every participant of the protocol can efficiently predict whether or not there will exist a block $B$ with $c_B = c$ such that $V_P(B)=1$, where $\predecessor^D(B) = A$ and $t_B  = t$.
\end{definition}

For the reader interested in further understanding predictability, Appendix~\ref{app:predict} contains some sample definitions for $V(\cdot)$ and analyzes their predictability.

\subsection{(Non)-Recency}
The second property we consider is recency, which is just the negation of local predictability. Intuitively, a protocol is $D$-recent at $A$ if the validity of block $A$ depends on some information contained in the last $D$ predecessors of $A$. The main security concern with $D$-recent protocols is that intuitively \emph{each chain has its own pseudorandomness} (but this is not a formal claim).\footnote{By \emph{chain}, we mean a set of blocks $\{\predecessor^i(B)\ |\ i \geq 0\}$.}  We'll again get into more detail with respect to security implications in Section~\ref{SEC:ATTACKS}, but just note here that certain deviations are easier to detect when chains share the same pseudorandomness. 

\begin{definition}[$D$-recent]
A coin $c$ is $D$-recent at a block $A$ for timestamp $t$ if the owner of $c$ \emph{cannot} efficiently predict whether or not there will exist a block $B$ such that $V_P(B)=1$, where $\predecessor^D(B) = A$, $c_B = c$, and $t_B  = t$. 
\end{definition}
\noindent As with predictability, in many existing protocols every coin will be $D$-recent at every block and for every timestamp. We will refer to such examples as $D$-recent protocols. Due to the following observation, further examples illustrating recency aren't necessary, as it's simply the negation of $D$-local predictability.

\begin{observation}
For any $D$, any block $A$ and any timestamp $t$, a coin $c$ is either $D$-locally predictable or $D$-recent.
\end{observation}

\section{Security Implications}\label{SEC:ATTACKS}
In this section, we elaborate on the security implications of predictability and recency.

\subsection{Global Predictability}
Here, we'll describe two attacks against protocols with coins that are globally predictable that we call ``globally predictable selfish-mine'' and ``globally predictable double-spend.'' In the former, the attacker attempts to claim extra mining rewards by delaying the announcement of mined blocks, and in the latter the attacker attempts to receive goods for free by overwriting a transaction in order to effectively spend the same coins twice. Both attacks are also possible, but weaker and a touch more complex, against locally predictable protocols. So we begin with the globally predictable versions, and defer the locally-predictable variants to Appendix~\ref{app:local}. Both attacks have a similar flavor, so we detail selfish mining here, and also defer double spending to Appendix~\ref{app:doublespend}.

\begin{definition}[Globally-Predictable Selfish Mining]  \hfill

\begin{enumerate}[leftmargin=*]
\item At all times $t$, let $A$ denote the current longest chain (that is, let $A$ be the block $B$ you are aware of maximizing $S(B)$).
\item For all $k > 0$, find the minimum time $t'_k$ such that there exists a block $B$, where $\predecessor^D(B) = A$ (for some $D > 0$), $V_P(B) = 1$, you own coin $c_{\predecessor^i(B)}$ for all $i \in [0,D-1]$, $t_B = t'_k$, and $S(B) > S(A)+k$. That is, for all $k$, find the earliest time that you can create a block $B$ with $S(B) > S(A)+k$, \emph{where you created all blocks on the path from $A$ to $B$}. 
\item For all $k >0$, find the minimum time $t^*_k$ such that there exists a block $B$, where $\predecessor^E(B) = A$ (for some $E > 0$), $V_P(B) = 1$, you \emph{don't} own coin $c_{\predecessor^i(B)}$ for all $i \in [0,E-1]$, $t_B = t^*_k$, and $S(B) > S(A)+k$. That is, for all $k$, find the earliest time that \emph{the rest of the network} can create a block $B$ with $S(B) >S(A)+ k$, \emph{and you did not create any blocks on the path from $A$ to $B$}. 
\item If at time $t$, there exists a $k$ such that $t'_k < t^*_k$, immediately stop publishing blocks until $t'_k$ (if there are multiple such $k$, take the largest one). At time $t'_k$, output the promised $B$ and its predecessors.
\end{enumerate}

\end{definition}
\vspace{-5mm}
\begin{figure}[ht]
\begin{center}
\begin{tikzpicture}[thick,scale=0.7, every node/.style={scale=0.7}]
    \node[main node,rectangle] (1) {$A$};
    \node[draw=none] (invt) [below = 0.1cm of 1] {$t$};
    \node[main node,rectangle] (2) [right = 1cm of 1]{$B_1$};
    \node[draw=none] (inv) [right = 1cm of 2] {$\dots$};
    \node[main node,rectangle] (3) [right = 1cm of inv]{$B_{k-1}$};
    \node[main node,rectangle] (4) [right = 1cm of 3]{$B$};
    \node[draw=none] (invt4) [below = 0.1cm of 4] {$t'_k$};
    
    \node[main node,rectangle] (5) [above = 0.8cm of 2]{$\hat{B}_1$};
    \node[draw=none] (inv2) [right = 0.8cm of 5] {$\dots$};
    \node[main node,rectangle] (6) [right = 0.8cm of inv2]{$\hat{B}$};
        \node[draw=none] (invt4) [below = 0.1cm of 6] {$S(\hat{B}) < S(B_k)$};

    \draw [dashed,decoration={markings,mark=at position 1 with {\arrow[scale=2,>=stealth]{>}}},postaction={decorate}] (2) to (1);
    \draw [dashed,decoration={markings,mark=at position 1 with {\arrow[scale=2,>=stealth]{>}}},postaction={decorate}] (inv) to (2);
    \draw [dashed,decoration={markings,mark=at position 1 with {\arrow[scale=2,>=stealth]{>}}},postaction={decorate}] (3) to (inv);
    \draw [dashed,decoration={markings,mark=at position 1 with {\arrow[scale=2,>=stealth]{>}}},postaction={decorate}] (4) to (3);
    \draw [decoration={markings,mark=at position 1 with {\arrow[scale=2,>=stealth]{>}}},postaction={decorate}] (5) to (1);
    \draw [decoration={markings,mark=at position 1 with {\arrow[scale=2,>=stealth]{>}}},postaction={decorate}] (inv2) to (5);
    \draw [decoration={markings,mark=at position 1 with {\arrow[scale=2,>=stealth]{>}}},postaction={decorate}] (6) to (inv2);
\end{tikzpicture}
\end{center}
\caption{Globally Predictable Selfish Mining (illustrated with $S(B) = $ number of predecessors of $B$): If you know you can make a node $B$ at time $t'_k$ whose predecessors after $A$ are also all created by you, and $S(B)>S(\hat{B})$ for all possible $\hat{B}$ that could be created by time $t'_k$ without you (whose predecessors after $A$ are also all created without you), stop publishing blocks and output the full chain $B_1 \rightarrow \ldots \rightarrow B$ all at once at time $t'_k$.} \label{fig:PSM}
\end{figure}

The high-level idea behind the attack is the following: the original selfish mining attack~\cite{eyal2014majority} proposes withholding a block $B$ upon creation (i.e. not broadcasting it). You continue mining on top of $B$, while the rest of the network continues mining on top of $\predecessor(B)$. If you create a new block $B'$ on top of $B$ before the rest of the network creates a new block on top of $\predecessor(B)$, then you now possess the unique longest chain. So you can continue mining on top of $B$ and its descendants until the rest of the network finds a chain that is almost as long as yours. At this point you can announce your chain, causing the entire chain built by the rest of the network to be orphaned (because they will all adopt your uniquely longest chain). Of course, the attack could go completely differently: maybe the rest of the network successfully mines on top of $\predecessor(B)$ before you mine on top of $B$. In this case, now you're in trouble and run the risk of losing $B$ because there is a competing chain of the same length. 

With sufficient global predictability, however, there is no risk! You can predict \emph{before deciding whether to withhold $B$} if you'll mine on top of $B$ before another miner mines on top of $\predecessor(B)$. So you can only withhold those $B$ for which the attack will succeed, completely avoiding the risk. 
For instance, if a protocol is $D$-globally predictable, and $S(B) = $\# predecessors of $B$, then the above attack can be carried out for any $k \leq D$. The attacker's incentives to carry out such an attack of course depend on exactly how minng rewards are distributed, but it is clear that  \emph{globally predictable selfish mining allows the attacker to produce a greater fraction of blocks on the longest chain.} For standard reward schemes this is indeed profitable~\cite{eyal2014majority,carlsten2016instability}.

The frequency with which an attacker will have the ability to predictably selfish-mine depends on the exact nature of the scoring function $S$ and the Proof-of-Stake protocol $P$. In Appendix~\ref{app:attacks}, we analyze the probability of a miner being able to launch a predictable selfish-mining or predictable double spend attack assuming that $S(\cdot)$ is the simple longest-chain rule and that $P$ acts as a random oracle (formal definition in Appendix~\ref{app:attacks} - without some assumption like this it's impossible to begin talking about probabilities). The key takeaway from this section is that \textbf{Predictable Longest-Chain Variant Protocols are vulnerable to Predictable Selfish Mining.}


\vspace{-3mm}

\subsection{Recency}
\label{sec:recency}
Here, we discuss an attack on $D$-recent protocols, which we call Undetectable Nothing-at-Stake. In the Nothing-at-Stake attack miners try to mine on top of many blocks simultaneously, instead of just the one maximizing $S(B)$ over all valid blocks. We call it undetectable if, information theoretically, there is no proof that a miner engaged in Nothing-at-Stake. To make this formal, it is helpful to first recall the behavior of an honest miner.
\vspace{-2mm}
\begin{definition}[Honest Miner] An \emph{Honest Miner} participating in a Proof-of-Stake protocol will do the following at every time step $t$:
\begin{itemize}[leftmargin=*]
\item Find $A$ maximizing $S(A)$ among all blocks that the miner is aware of.
\item For all owned coins $c$, attempt to mine a new block $B = M_P(A,c,t)$. If $B \neq \bot$ announce the new block $B$, otherwise do nothing.
\end{itemize}
\end{definition}

Any deviation from this behavior would be considered dishonest, and problematic for the functionality of the underlying protocol. Sometimes, these deviations will be \emph{detectable}, in the sense that there is clear evidence that a miner deviated from the protocol. Formally:

\begin{definition}[Provable Deviation]
We say that two valid blocks $B$ and $B'$ with $c_B = c_{B'} = c$ are a \emph{provable deviation} by the owner of coin $c$ if $t_{B'} = t_B$, or if both $t_{B'} > t_B$ and $S(B) > S(\predecessor(B'))$. 
\end{definition}

The first part of this definition captures that honest nodes only output one block per timestep. The second half captures that at time $t_B$, the owner of coin $c$ created block $B$. Then at time $t_{B'} > t_B$, they are claiming that $\predecessor(B')$ maximizes $S(\cdot)$ over all blocks they are aware of. Clearly this is not true if $S(B) > S(\predecessor(B'))$. The proof of Proposition~\ref{prop:provable} is in Appendix~\ref{app:recent}.

\begin{proposition}\label{prop:provable}
If a miner is caught having announced a provable deviation, then they must have deviated from the intended protocol. Also, any sequence of announcements from a miner that does \textbf{not} contain a provable deviation could have been sent by an honest miner experiencing latency. 
\end{proposition}

Proposition~\ref{prop:provable} tells us that without a provable deviation, we cannot ``punish'' suspected deviant miners without the risk of punishing honest but poorly-connected miners. Now, we propose one specific protocol deviation that is guaranteed never to produce a provable deviation.

\begin{definition}[Undetectable Nothing-at-Stake] First, ensure that all of your coins are owned by different public keys. Then, for each owned coin $c$, do the following during every timestep $t$:
\begin{itemize}[leftmargin=*]
\item Find $A$ maximizing $S(A)$ among all blocks that you are aware of.
\item Find $A'$ maximizing $S(A')$ among all blocks that are \textbf{not} descendants of $\predecessor^{D}(A)$.
\item Mine $B = M_P(A,c,t)$ with $\predecessor(B) = A$, and $B' = M_P(A',c,t)$ with $\predecessor(B') = A'$. 
\item If $B \neq \bot$, announce the new block $B$. If $B' \neq \bot$, \textbf{and announcing $B'$ would not create a provable deviation}, announce $B'$. 
\end{itemize}
\end{definition}

\begin{figure}[ht]
\begin{center}
\begin{tikzpicture}[thick,scale=0.7, every node/.style={scale=0.7}]
    \node[main node,rectangle] (1) {};
    \node[main node,rectangle] (2) [right = 1cm of 1]{$\predecessor^{D}(A)$};
    \node[draw=none] (inv) [right = 1cm of 2] {$\dots$};
    \node[main node,rectangle] (3) [right = 1cm of inv]{$\predecessor(A)$};
    \node[main node,rectangle] (4) [right = 1cm of 3]{$A$};
    \node[main node,rectangle] (5) [right = 1cm of 4]{$B$};
    \node[main node,rectangle] (6) [above right = 0.5cm and 1cm of 1]{};
    \node[draw=none] (7) [right = 1cm of 6] {$\dots$};
    \node[main node,rectangle] (8) [right = 1cm of 7]{$A'$};
    \node[main node,rectangle] (9) [right = 1cm of 8]{$B'$};

    \draw [decoration={markings,mark=at position 1 with {\arrow[scale=2,>=stealth]{>}}},postaction={decorate}] (2) to (1);
    \draw [decoration={markings,mark=at position 1 with {\arrow[scale=2,>=stealth]{>}}},postaction={decorate}] (inv) to (2);
    \draw [decoration={markings,mark=at position 1 with {\arrow[scale=2,>=stealth]{>}}},postaction={decorate}] (3) to (inv);
    \draw [decoration={markings,mark=at position 1 with {\arrow[scale=2,>=stealth]{>}}},postaction={decorate}] (4) to (3);
    \draw [dashed,decoration={markings,mark=at position 1 with {\arrow[scale=2,>=stealth]{>}}},postaction={decorate}] (5) to (4);
    \draw [decoration={markings,mark=at position 1 with {\arrow[scale=2,>=stealth]{>}}},postaction={decorate}] (6) to (1);
    \draw [decoration={markings,mark=at position 1 with {\arrow[scale=2,>=stealth]{>}}},postaction={decorate}] (7) to (6);
    \draw [decoration={markings,mark=at position 1 with {\arrow[scale=2,>=stealth]{>}}},postaction={decorate}] (8) to (7);
    \draw [dashed,decoration={markings,mark=at position 1 with {\arrow[scale=2,>=stealth]{>}}},postaction={decorate}] (9) to (8);
\end{tikzpicture}
\end{center}
\caption{Undetectable Nothing-at-Stake: The attacker creates two blocks, $B$ and $B'$. There is no resource cost in checking if both are valid. Depending on reward scheme, there may be an expected monetary gain for announcing both.}\label{fig:UNaS}
\end{figure}
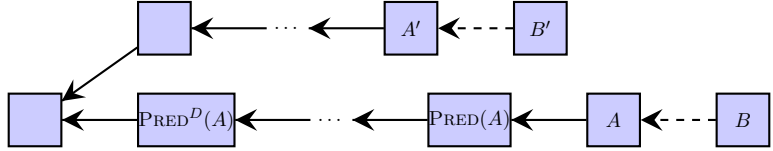

By definition, announcing $B$ and $B'$ at time $t$ does not create a provable deviation at time $t$. But, it could be the case that at some later time $t' > t$, announcing some new block $C$ would reveal that the miner was not following the honest protocol. We show next that this is impossible.

\begin{observation}\label{obs:undetectable}
The Undetectable Nothing-at-Stake strategy never produces a provable deviation. 
\end{observation}

So the proposed deviation will never ``get caught'' (that doesn't necessarily mean that a clever protocol can't still enact punishment - see Section~\ref{sec:discussion}). But we also want to understand whether the proposed deviation will ever actually deviate (it's conceivable that the safety check will prevent the miner from ever announcing an ``illegitimate'' block), and this is where recency comes in. Essentially what's going on is that if a protocol is $D$-Recent, then whether or not a you can build a valid block on top of $A'$ with coin $c$ at time $t$ actually depends on some of the blocks between $A'$ and $\predecessor^D(A)$ (which, by Chain Dependence + Monotonicity, the validity of any block built on top of $A$ \emph{doesn't} depend on). So each coin $c$ is kind of getting a ``fresh shot'' at being eligible to mine a block on top of $A'$ during time $t$, and one might reasonably expect this shot to succeed with non-zero probability (with the success probability of course dependent on the exact behavior of $V_P(\cdot)$). 

Again, any meaningful probabilistic analysis requires some assumption on $S$ and $P$ (otherwise we don't even have a probability space to work with). We consider the case where $S$ is the simple longest chain rule and $P$ acts as a random oracle, and show essentially that when the number of coins in the system is larger than the recency of the protocol, then Undetectable Nothing-at-Stake is announcing twice as many blocks as the honest strategy. See Appendix~\ref{app:attacks} for a formal statement. 

Again, the profitability of Undetectable Nothing-at-Stake depends on the exact reward scheme, but what is clear is that Undetectable Nothing-at-Stake allows the attacker to (undetectably) produce a greater fraction of blocks. The main takeaway from this section is that \textbf{Recent Longest-Chain Variant Protocols are vulnerable to Undetectable Nothing-at-Stake}. 




\section{Discussion}
In summary, we've shown that even in an ideal model with perfect connectivity and no latency, every Longest-Chain-Variant-Proof-of-Stake protocol has some undesirable property. Below, we discuss possible fixes for the attacks enabled by these properties, draw conclusions and pose directions for future work.
\paragraph{Preventing predictable selfish mining.} Preventing predictable selfish mining is challenging, but some clever ideas exist in the literature. At a high level, Fruitchains~\cite{pass2016fruitchains}, Ouroboros~\cite{kiayias2017ouroboros}, and Tezos~\cite{Tezos} design protocols where blocks need to be ``supported'' once mined (eligibility to support is also proportional to stake), so one would not only need a majority of blocks mined but also a majority of ``support tokens(/fruit)'' in a given window to successfully selfish mine. Both Snow White \cite{bentov2016snow} and Ouroboros \cite{kiayias2017ouroboros} provide proofs that any deviation from their prescribed protocol can only provide a small $\varepsilon$ in additional mining rewards. However,~\cite{bentov2016snow} notes that it would be preferable for known attacks to be \emph{strictly} disincentivized (more on this in Appendix~\ref{sec:examples}), and it remains open whether these reward schemes accomplish this.

\paragraph{Preventing predictable double-spend.} A simple defense specifically against predictable double-spend attacks is to accept long confirmation times (e.g. a transaction is not considered ``finalized'' by vendors until several blocks have been announced descending from the block containing the transaction). Our analysis (Appendix~\ref{app:attacks}) indicates that several hundred blocks might be necessary, making it virtually impossible to have quick confirmation times in a predictable Proof-of-Stake protocol without further defenses. For example, the authors of the  Ouroboros~\cite{kiayias2017ouroboros} Protocol (which is predictable) suggest using confirmation times of 148 minutes to defend against double spend attacks by an attacker controlling $40\%$ of the stake when blocks are created at a rate of one per minute (and this is consistent with our analysis).



\paragraph{Defending against Undetectable Nothing-at-Stake.} Existing literature proposes roughly three paradigms that attempt to defend against Undetectable Nothing-at-Stake. The most common defense in commercial protocols is to set $D$ very large (these protocols are therefore $D$-locally predictable for large $D$), and to use some form of ``checkpointing'' every $\leq D$ blocks. This ``checkpointing'' might be run externally by a trusted party, hard-coded into the protocol, or just a form of trust among network participants that they would never seriously consider a fork more than $D$ blocks back. In practice, there don't seem to have been any serious issues with this approach, but to our knowledge its security hasn't previously been rigorously analyzed.

Algorand proposes a different approach: instead of using a longest-chain variant, it uses a Byzantine consensus protocol. Under some network connectivity assumptions, they show that the probability of a fork is negligible. As such, any deviant behavior that results in a fork (such as Undetectable Nothing-at-Stake) can be readily recognized as malicious, and safely ignored. 

Ethereum's Casper~\cite{casper} proposes a third solution that they call ``dunkles'': punish every miner whose block winds up being orphaned (not a predecessor of the block maximizing $S(A)$). The high-level goal of this is to essentially copy the incentives from Proof-of-Work: if your block is orphaned, you still lose the electricity that went into mining it. 
By punishing the miner of every orphaned block, some honest miners will get punished just by bad luck, but it will also discourage attackers from mining off the longest chain. This seems like a promising direction, but there is currently no formal specification or rigorous evaluation of the proposal.

\paragraph{There's no reason for global predictability.} There's a real tradeoff to explore between $D$-local predictability and $D$-recency (since one cannot avoid both, by definition). But there doesn't seem to be any benefit to global predictability, only the risk of stronger predictable-selfish-mining attacks. There also doesn't seem to be a black-box reduction stating that mixing any existing protocol properly with digital signatures removes global predictability (although this would be a great future result), but it seems likely that a clever use of digital signatures as in Algorand could modify most existing protocols to be no longer globally predictable with low cost. 

\paragraph{Trusted External Randomness, Trusted Checkpointing, or Not?} Our work shows that there is a fundamental difference between trusted external randomness and ``internal pseudorandomness'' derived pseudorandomly using the cryptocurrency itself. Essentially, the difference is that all internal pseudorandomness is in the end locally predictable or recent and therefore susceptible to some form of attack. External randomness (such as the NIST beacon) is not, and recall that protocols that use external randomness fail to satisfy Chain Dependence (so none of our results apply). 
Can such a protocol be the basis for a secure, incentive compatible Proof-of-Stake protocol? Even if the answer is yes, how does one resolve the (possibly just-as-challenging) issue of obtaining trusted randomness?\footnote{Chapter 9.4 of~\cite{narayanan2016bitcoin} describes how functional cryptocurrencies can serve as sources of trusted randomness. So it is indeed unclear whether trusted randomness is really an easier problem.} 
Trusted checkpointing provides a cheap solution to Undetectable Nothing-at-Stake, and does seem simpler than trusted randomness.  
Can trusted checkpointing provide guarantees that are otherwise hard (or impossible) to come by?

\paragraph{Byzantine Consensus versus Longest-Chain Variants.} With the exception of Algorand and Casper: the Friendly Finality Gadget, every proposal that the authors are aware of uses a longest-chain variant. The upside of longest-chain variants is that they are inherently robust to Eclipse attacks. Byzantine consensus protocols require some network connectivity assumptions in order to safely ignore messages sent too far in the past, and are less robust to Eclipse attacks. Are there provable limits to what can be achieved by longest-chain variants, necessitating the use of Byzantine Consensus? Or is it possible to achieve the same guarantees with a true Longest-Chain variant?

\paragraph{Rigorously and Transparently Evaluate Protocols in the Ideal Model.} Our work shows that already it is quite challenging to design incentive compatible Longest-Chain Variants in the ideal model. Numerous reasonable ideas have been proposed to address the vulnerabilities arising from predictability and recency, but none have transparent yet rigorous analysis. We believe that the vulnerabilities we've uncovered are serious enough that future Proof-of-Stake proposals should include transparent proofs of how they defend against predictable double-spend, predictable selfish-mine, and undetectable nothing-at-stake in the ideal model. 

It is obviously necessary to also continue evaluating network security aspects of proposed protocols, but it is important that any incentive-related security claims in these rich models easily map to transparent claims in the ideal model.

\label{sec:discussion}

\bibliographystyle{plainnat}
\bibliography{references}

\appendix
\section{Examples Aiding Definitions From Section~\ref{sec:prelim}}\label{app:prelim}
\subsection{Assumption~\ref{assumption:protocol}}\label{app:assumption}
To get some intuition for Assumption~\ref{assumption:protocol}, here are two toy examples that violate Chain Dependence and Monotonicity (respectively). First, consider a protocol where the validity of a block depends on a trusted source of external randomness (e.g. the NIST beacon). This external randomness is not contained in the blockchain itself, so such a protocol does not satisfy Chain Dependence. 

Next, consider a protocol that declares a block $B$ invalid if its creator proposed another block $B'$ within the same $2$-week time period. Then $B$ will be considered valid for the graph that contains only $B$ and its predecessors, but invalid for the graph which contains $B, B'$ and all of their predecessors. So this protocol violates Monotonicity. These protocols are not absurd, but are much more vulnerable to Eclipse attacks than virtually all existing proposals.

\subsection{Freezing}\label{app:freeze}
In the definition for Proof-of-Stake we provided, it may sometimes be desirable to ``freeze'' coins used to mine for longer than just one block. That is, the first indicator in bullet three requires that coin $c$ be owned by $\miner(B)$ in the block before $B$ (so that the miner could not move $c$ in the same block s/he is trying to mine). For many protocols, it may be desirable to additionally insist that the owner of coin $c$ did not change in any of the last $F$ blocks (i.e. that $c$ was \emph{frozen} for $F$ blocks before used for mining). We note that our definition can easily be modified to replace $\mathbb{I}\{ \owner_{\predecessor(B)}(c_B) = \miner(B) \}$ with $\mathbb{I}\{ \owner_{\predecessor^i(B)}(c_B) = \miner(B) \ \forall i \leq F\}$, where $F$ is a freezing parameter of $P$ (and Definition~\ref{def:PoS} is then a special case with $F = 1$ hardcoded). Of course, this could also be offloaded into $V_P(\cdot)$, but the astute reader will later notice that treating freezing separately makes future technical definitions cleaner. 

\subsection{Predictability}\label{app:predict}
\noindent Below we provide examples of predictable and unpredictable protocols for the sake of further explanation. In what follows, let $A = \predecessor^D(B)$ and $T \in \reals$ be some positive threshold. 

\begin{itemize}[leftmargin=*]
\item Protocol $P_1$ where $V_{P_1}(B) = 1$ if and only if $\sha(A, t_B, c_B) <T$.\footnote{Note that this is (essentially) the initial proposal made in Peercoin: \url{https://peercoin.net/}.} Every coin $c$ in $P_1$ is $D$-globally predictable at $A$, for all $D$ and all $A$. This is because every user can compute $\sha(\predecessor^{D-1}(A), t', c')$ for all $t', c'$ and check whether or not it's $<T$. So every user can compute the minimum $t'_1$ such that a block can be built on top of $A$. Similarly, every user can compute $\sha(\predecessor^{D-i}(A), t', c')$ for all $i \leq D$, which determines the minimum $t'_i$ that a block can be built with $\predecessor^{i}(B) = A$ (note that $t'_i$ must be monotonically non-decreasing in $i$). If $t'_D < t$, then a block $B$ indeed exists with $\predecessor^D(B) = A$, $c_B = c, t_B = t$, $V_{P_1}(B) = 1$. Otherwise, no such $B$ exists. 

\item Protocol $P_2$ where $V_{P_2}(B) = 1$ if and only if $\sha(t_B, c_B) <T$. Every coin $c$ in $P_2$ is $D$-globally predictable at $A$ for all $D$ and all $A$. This is because every user can compute $\sha(t', c')$ for all coins $c'$ and times $t'$. Therefore, every user can determine all potential timesteps where a block could be created. If there are $\geq D$ such timesteps between $t_A$ and $t$, and $\sha(t, c) < T$, then the answer is yes. If not, then the answer is no. 

\item Protocol $P_3$ where:\footnote{$P_3$ is based on the seed-selection portion of Algorand.}
\begin{itemize}
\item Each block $B$ contains a signature $s_B$ computed by $\miner(B)$.
\item $M_{P_3}(B,c,t)$ outputs a block $B$ with $s_B = \text{SIG}_{\owner(c)}(\sha(s_A),t)$, where $\text{SIG}_{\owner(c)}(\cdot)$ denotes the function which digitally signs a message using the secret key of $\owner(c)$.
\item $V_{P_3}(B) = 1$ if and only if $\sha(s_B) < T$.
\end{itemize}
$P_3$ is a good example to clarify potentially subtle aspects of the definitions. Every coin in $P_3$ is $1$-locally predictable at $A$, for all $A$ (as with all Proof-of-Stake protocols). In addition, $\owner(c)$ can computationally efficiently find certain kinds of blocks $B$ with $t_B = t$ and $\predecessor^D(B) = A$: namely, those for which $\miner(B') = \owner(c)$ for all $B' = \predecessor^i(B)$ for $i \leq D$. This is because $\owner(c)$ can efficiently compute all blocks $B'$ that they themselves can build on top of $A$, and then all blocks that they themselves can build on top of these blocks, etc. using $M_{P_3}$. However, $\owner(c)$ cannot computationally efficiently predict whether there exists a block $B$ with $\predecessor^D(B) = A$, $t_B = t$, but $\miner(\predecessor^i(B)) \neq \owner(c)$ for some $i < D$. This is because knowing the existence of this block would require being able to digitally sign as $\miner(\predecessor^i(B))$, which $\owner(c)$ cannot do computationally efficiently. So every coin $c$ in $P_3$ is \emph{not} $D$-locally predictable for any $D > 1$. Similarly, no coin is $D$-globally predictable for any $D$ because in order to know whether $B$ is valid, one must be able to digitally sign messages as $\owner(c_B)$ (which one cannot do computationally efficiently unless one is $\owner(c_B)$).
\end{itemize}

\section{Omitted Proofs from Section~\ref{SEC:ATTACKS}}\label{app:attacks}
\subsection{Locally Predictable Selfish Mining}\label{app:local}
Here, we'll show how to modify the globally predictable attacks of the previous section to be locally predictable. Locally predictable selfish mining is no longer risk-free (because you can't predict when the rest of the network will find their blocks), but you can still gain a statistical edge by knowing when in the future your blocks will come (essentially, if your blocks come earlier than normal, this is a good time to withhold. If your blocks come later than normal, this is a bad time).  

\begin{definition}[Locally-Predictable Selfish Mining]  \hfill
\begin{enumerate}[leftmargin=*]
\item For all $k>1$, define a time cutoff $T_k$.
\item For all $k > 1$, find the minimum time $t'_k$ such that there exists a block $B$, where $\predecessor^D(B) = A$ (for some $D > 0$), $V_P(B) = 1$, you own coin $c_{\predecessor^i(B)}$ for all $i \in [0,D-1]$, $t_B = t'_k$, and $S(B) > S(A)+ k$. That is, for all $k$, find the earliest time that you can create a block $B$ with $S(B) > S(A)+ k$, \emph{where you created all blocks on the path from $A$ to $B$}. 
\item If at time $t$, there exists a $k$ such that $t'_k \leq t+T_k$, immediately stop publishing blocks until $t'_k$ (if there are multiple such $k$, take the largest one). At time $t'_k$, output the promised $B$ and its predecessors.
\end{enumerate}



\end{definition}

As referenced above, the key difference between Locally-Predictable and Globally-Predictable Selfish Mining is that you can no longer compute how long it will take for the rest of the network to produce a block with score $S(A) + k$. Still, you can get a statistical edge because you know at what time in the future you'll be able to produce a block with score $S(A) + k$. So set the cutoff $T_k$ so that you will actually gain in expectation by withholding. Note that there certainly exists such a $T_k$ (e.g. $T_k = t$), although the probability of producing blocks before $T_k$ might be extremely small (e.g. for $T_k = t$ it is zero). Still, for all existing Longest Chain variants that fit our framework, for all $\alpha > 0$, a user with an $\alpha$-fraction of the total stake could set appropriate thresholds for Locally Predictable Selfish Mining and produce a $>\alpha$-fraction of the total blocks on the longest chain (and also strictly increase their expected reward).

Moreover, since Locally Predictable Selfish Mining only requires predicting your own blocks, $1$-Local Predictability actually suffices for this attack. By Observation~\ref{obs:1locallypredictable}, this means that in fact \emph{every Longest-Chain Variant Proof-of-Stake Protocol is vulnerable to Locally Predictable Selfish Mining}. 
\begin{lemma}
Every Proof-of-Stake protocol $P$ is vulnerable to the Locally Predictable Selfish Mining attack. In particular, every time a miner is eligible to mine a block, she can also attempt to launch a Predictable Selfish Mining attack (for all $k$).
\end{lemma}
\begin{proof}
By Observation \ref{obs:1locallypredictable}, for any block $B$, any coin $c$ owned by the attacker, and any time $t$, the attacker can determine whether or not there exists a block $B'$ with $\predecessor(B') = B$, $c_{B'} = c$, $t_{B'} = t$ and $V(B') = 1$. Therefore, the attacker can do the following for all $k$:
\begin{itemize}[leftmargin=*]
\item Initialize $\mathcal{B} = \{A\}$. 
\item While there exists a block $B$ with $c_B$ owned by the attacker, $t_B < t + T_k$, $\predecessor(B)\in \mathcal{B}$ (in English: while the attacker can mine a block on top of some block in $\mathcal{B}$ at time $< t + T_k$): Add to $\mathcal{B}$ all such blocks.
\item Let $B_k = \arg\max_{B \in \mathcal{B}}\{S(B)\}$. If $S(B_k) > S(A) + k$, then we've found an opportunity to selfish mine. If not, then there's no opportunity.  
\end{itemize}

Note that the attacker can implement every step above due to Observation~\ref{obs:1locallypredictable} and the prior reasoning, and it's trivial to see that the algorithm above implements locally predictable selfish mining.
\end{proof}

\paragraph{Minor Improvements with Greater Local Predictability.} If a protocol happens to have $D$-locally predictable coins for larger $D$ instead of just $1$-locally predictable, then an attacker intending to launch a predictable selfish mining attack is aware $D$ blocks in advance. It's unclear that this advanced notice is significant, but it's not completely negligible. For example, a miner who can predict that she will likely succeed with a $k>6$ locally predictable selfish mine several hours in the future may offer to accept bribes in order to fork for a double-spend attack (similarly, the miner could try to prepare their own double-spend attack - see Appendix). The idea is that if the goal of this attack is simply to get increased mining rewards, then the advanced notice doesn't help. But if the goal is to use this attack in more ``outside-the-box'' ways (or to do a predictable double-spend), then the advanced notice might actually help.

The key takeaway from this subsection is that \textbf{Every Longest-Chain Variant Proof-of-Stake Protocol is vulnerable to Locally Predictable Selfish Mining}. Again, the improved rewards for participating in the attack vary from protocol to protocol, but in all existing protocols that the authors are aware of, the reward increase is non-zero. There are interesting ideas for potential defenses posed in both commercial and academic protocols, but without transparent analyses.

\subsection{Predictable Double-Spend}\label{app:doublespend}
We first describe a ``predictable double-spend,'' which requires an initial definition of a ``confirmation time.''
Confirmation times aren't hard-coded into cryptocurrencies, but determine when a vendor is comfortable considering a transaction ``finalized'' and exchanging goods. The required confirmation may vary depending on the transaction. For instance, a cafe may be willing to hand over a cup of coffee even before the payment transaction has been included in a Bitcoin block (but at least verifying that the transaction has been digitally signed and broadcast). But a homeowner may not hand over the deed to their house until several Bitcoin blocks have been mined on top of the block containing the payment transaction. 

\begin{definition}[Confirmation Time] For a given block $B$ containing transaction $x$, and block $B'$ a descendant of $B$, we say that $x$ is \emph{confirmed} by $B'$ if a vendor would exchange whatever goods are being purchased by $x$ once believing (according to whatever Longest-Chain-Variant is used) the history defined by $B'$. 
\end{definition}

Below we now describe a predictable double-spend attack which is costless to the attacker, and may result in the attacker receiving goods for free if confirmation times are lax enough.

\begin{definition}[Globally Predictable Double-Spend] Do the following:
\begin{enumerate} 
\item Produce a transaction $x$ in order to purchase some good, but don't yet announce it to the network or vendor.
\item At all times $t$, let $A$ denote the current longest chain (that is, let $A$ be the block $B$ you are aware of maximizing $S(B)$).
\item For all $k > 0$, find the minimum time $t'_k$ such that there exists a block $B$, where $\predecessor^D(B) = A$ (for some $D > 0$), $V_P(B) = 1$, you own coin $c_{\predecessor^i(B)}$ for all $i \in [0,D-1]$, $t_B = t'_k$, and $S(B) > S(A)+k$. That is, for all $k$, find the earliest time that you can create a block $B$ with $S(B) > S(A)+k$, \emph{where you created all blocks on the path from $A$ to $B$}. 
\item Similarly, for all $k >0$, find the minimum time $t^*_k$ such that there exists a block $B$, where $\predecessor^E(B) = A$ (for some $E > 0$), $V_P(B) = 1$, you \emph{don't} own coin $c_{\predecessor^i(B)}$ for all $i \in [0,E-1]$, $t_B = t^*_k$, and $S(B) > S(A)+k$. That is, for all $k$, find the earliest time that \emph{the rest of the network} can create a block $B$ with $S(B) >S(A)+ k$, \emph{where you did not create any blocks on the path from $A$ to $B$}. 
\item If at time $t$, there exists a $k$ such that $t'_k < t^*_k$, immediately announce the transaction $x$ and stop publishing blocks until $t'_k$ (if there are multiple such $k$, take the largest one).

\item Hope that at some $t'' < t'_k$, a block $B'$ is announced that confirms $x$.
\item Get the good from the vendor at time $t'$.
\item  At time $t'_k$, output the promised $B$ and its predecessors.
\end{enumerate}
\end{definition}

It should be clear that if this attack is successful, it will result in the attacker getting their goods for free, as $B$ will become the new history and no history built upon $B$ can possibly contain the transaction $x$. Notice that the above description is essentially the same as the globally predictable selfish-mining attack, but the goal of predictable selfish mining is to claim extra mining rewards rather than free goods for a canceled transaction. For locally predictable protocols, there is a corresponding locally predictable double-spend attack that is identical to the locally predictable selfish mining attack, except for the addition of the steps announcing the transaction $x$ which the attacker intends to cancel.

Let's see how this plays out with a relevant example. To do any meaningful analysis, we'll want to restrict attention to protocols with mining functions that \emph{act as a random oracle} (otherwise we can't even begin to talk about probabilities). By this, we mean protocols $P$ where no matter how many times $M_P$ has already been queried on other inputs, querying $M_P$ on fresh input appears to be an independent random variable (this is a standard cryptographic assumption when discussing ideal hash functions). 

\begin{definition}[Random Oracle] 
We say that a mining function $M_P$ \emph{acts as a random oracle} with recency $\ell$ if there exists a function $M^*$ such that for all blocks $B$, coins $c$ and timestamps $t$:
\begin{itemize}
\item $M_P(B, c, t) = M^*(\predecessor^\ell(B), c, t)$ 
\item There exists a success probability function $s(\cdot)$ such that $M^*(\predecessor^\ell(B), c, t) \neq \bot$ with probability $s(\predecessor^\ell(B))$.
\item For all $(B_1, t_1, c_1),\ldots,(B_k,t_k, c_k)$ the set $\{M^*(B_1,c_1,t_1),\ldots M^*(B_k,c_k,t_k)\}$ are independent random variables.\footnote{Obviously no protocol actually achieves this quality of randomness, but all existing proposals use some proxy such as a call to \sha ~or a ``Follow-the-Satoshi''~\cite{bentov2016cryptocurrencies,Tezos}.}
\end{itemize}
\end{definition}

So let's now consider Bitcoin's canonical longest-chain variant: $S(B) = \min\{\ell, \predecessor^\ell(B) = \bot\}$, and a protocol where $M_P$ acts as a random oracle with recency $D$. This means that the potential attacker is capable of computing the following two quantities correctly with high probability:
\begin{itemize}
\item For all $\ell \leq D$, what is the minimum $t^*$ such that there exist $t < t_1,\ldots, t_\ell = t^*$ and coins $c_1,\ldots, c_\ell$ \textbf{all not owned by the attacker}, and $M_P(\predecessor^{D-\ell}(B), c_i, t_i) \neq \bot$. Call this number $H_B(\ell)$. 
\item For all $\ell \leq D$, what is the minimum $t^*$ such that there exist $t < t_1,\ldots, t_\ell = t^*$ and coins $c_1,\ldots, c_\ell$ \textbf{all owned by the attacker}, and $M_P(\predecessor^{D-\ell}(B), c_i, t_i) \neq \bot$. Call this number $A_B(\ell)$. 
\end{itemize}

That is, for every $\ell \leq D$ the attacker makes the following thought experiment: ``Can I create $\ell$ blocks faster than the rest of the users?''. If $A_B(\ell) < H_B(\ell)$, i.e. the answer to the previous question is ``Yes'', the attacker can attempt to launch a predictable selfish-mine or predictable double-spend attack (but of course, whether or not the attacks achieve the attacker's goal depends on further details of the protocol). Note that in the setting where every coin is $D$-locally predictable at all blocks $A$ and timestamps $t$, the attacker can compute $A_B(\ell)$ but not $H_B(\ell)$. However, if $M_P$ acts as a random oracle, then the attacker can estimate the value of $H_B(\ell)$ with high probability. If every coin is $D$-globally predictable, then the attacker can compute $H_B(\ell)$ with probability 1.

Now the first question one might wish to ask is: given that the attacker controls an $\alpha$-fraction of the stake, what is the probability that we will ever see $A_B(\ell) < H_B(\ell)$ for any block before the end of the universe (a similar question was addressed by figures in~\cite{kiayias2017ouroboros})? Assuming that $M_P$ acts as a random oracle (with any recency), the probability that $A_B(\ell) < H_B(\ell)$ can be analyzed in the following way. $2\ell-1$ biased coins are flipped; the probability of heads is $\alpha$. 
Notice that every outcome of this experiment will have either at least $\ell$ heads or at least $\ell$ tails, but not both. Therefore, the probability that we see at least $\ell$ heads is exactly the same as the probability that the attacker creates $\ell$ blocks faster than the rest of the users. The random variable of interest follows the binomial distribution with parameters $\alpha$ and $2\ell - 1$, so the calculation is straightforward. Let $p_{\alpha,\ell} = Pr[ \text{attacker with $\alpha$ fraction of stake wins the race of $\ell$ blocks} ]$. We say that $\ell$ is \textit{safe} if $p_{\alpha,\ell} < T$ for some tolerance threshold $T$. For every such threshold $T$ we can find the smallest safe $\ell^*_{\alpha,T}$, i.e. the smallest window in which an adversary can win the race against the rest of the users with probability within the threshold. Intuitively, we would like to set a threshold so that $\ell^*_{\alpha,T}$ is small. 

The next question is how to set this threshold. If a block is created every minute, a cryptocurrency that lasts $1000$ years has approximately $5*10^8$ blocks. If we want the probability that the attacker succeeds at some time during the lifetime of the currency to be at most $10^{-7}$, the threshold $T$ should be set to $2 \cdot 10^{-16}$. Given $T$, we can plot $\ell^*_{\alpha,T}$ as a function of $\alpha$. See Figure~\ref{fig:a vs min safe ell}.
\begin{figure}[htbp]
	\begin{center}  
		\includegraphics[scale=0.7]{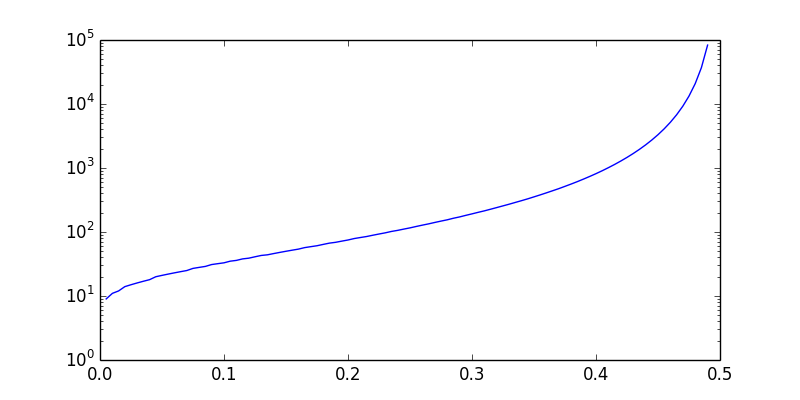}
	\end{center}
	\caption{$\alpha$ and $\ell^*_{\alpha,T}$ on a log-scale}\label{fig:a vs min safe ell}
\end{figure}
So if, for instance, one is comfortable assuming that no miner will exceed 40\% of the total stake, one can reasonably expect to never see $A_B(816) < H_B(816)$ for any block $B$ throughout the cryptocurrency's entire lifetime. 


\subsection{Recency}\label{app:recent}
\begin{proof}[Proof of Proposition~\ref{prop:provable}]
The first part of the proposition, i.e. provable deviations can only be created by miners who actually deviate from the honest protocol, is obvious. We show that any sequence of blocks that do not contain a provable deviation could conceivably have been created by a miner honestly following the protocol. Consider a set of valid blocks, $B_1,\ldots, B_i,\ldots$, sorted in increasing order of $t_{B_i}$, and with $c_{B_i} = c$ for all $i$, that contains no provable deviation (this immediately implies that $t_{B_i} < t_{B_{i+1}}$ for all $i$). If there are no provable deviations, this also means that $S(B_i) \leq S(\predecessor(B_{i+1}))$ for all $i$. Now, consider an honest miner who owns coin $c$ and no other coins. It is possible that at time $t_{B_i}$, the miner was aware of blocks of the form $S_i = \{\predecessor^\ell(B_j), \ell \geq 0, j <i\}\cup \{\predecessor^\ell(B_i), \ell > 1\}$, as these sets are monotone increasing.\footnote{Of course, it is not possible in the ideal network model if all other miners are participating honestly. But it is certainly possible if any other miner is dishonest, or our strong ideal network conditions are violated} Moreover, by hypothesis that the blocks contain no provable deviations, we also have that $\arg\max_{A \in S_i}\{S(A)\} =\predecessor(B_i)$ (because the miner also \emph{must} have been aware of all of these blocks in order to mine block $B_i$ at time $t_{B_i}$).\footnote{If there is a tie, have the miner tie-break in favor of $\predecessor(B_i)$. This is consistent as $\predecessor(B_i) \notin S_{i-1}$.} So it is entirely possible (although perhaps unlikely, depending on the exact protocol and behavior of other miners) that the owner of coin $c$ was aware of exactly the blocks $S_i$ at time $t_{B_i}$ and therefore produced their proofs by honestly following the protocol.
\end{proof}

\begin{proof}[Proof of Observation~\ref{obs:undetectable}]
First, it is clear that announcing the block $B$ with $\predecessor(B) = A$, where $A$ maximizes $S(A)$ over all blocks you are aware of cannot possibly create a provable deviation. This is because the block $B$ you just announced necessarily has $S(A) \geq S(A')$ and $t_B > t_{C'}$ for any other $C'$ you previously produced (by definition of $B$). So announcing a block on top of the longest chain can never produce a provable deviation. Moreover, there is an explicit check before announcing any $B'$ which doesn't build on top of the longest chain to guarantee that it also doesn't cause a provable deviation: simply check that $S(\predecessor(B')) > S(C)$, where $C$ is the most recent block you have mined.
\end{proof}

In this section we provide a formal analysis of how much of an advantage can be won by an attacker using the Undetectable Nothing-at-Stake strategy in a protocol with $D$-recent coins.

As a helpful (non)-example, consider the following protocol which is \emph{not} $D$-Recent for any $D$: $V_P(B) = 1$ if and only if $\sha(t_B, c_B) < T$. Now, for every pair of candidate blocks $C$ and $C'$ with $c_C = c_{C'}$, we have $V_P(C) = 1 \Leftrightarrow V_P(C') = 1$. In this case, announcing $C'$ together with $C$ will \emph{always} create a provable deviation, and the Undetectable Nothing-at-Stake behaves exactly as an honest miner. This is essentially because in the above protocol the blocks $C$ and $C'$ on separate forks use an identical source of pseudorandomness to determine if the owner of a given coin should be allowed to mine.

The problem with $D$-Recent protocols that act as random oracles with recency $\ell < D$ is that the two sides of a fork will have independent pseudorandom seeds, so the safety check rarely prevents the attack from announcing the illegitimate block. We will need one more definition to describe the full generality in which the following analysis holds.

\begin{definition}[More-Than-Honest] We say that a strategy is \emph{More-Than-Honest} if at every timestep $t$:
\begin{itemize}
\item Let $B$ maximize $S(B)$ over all blocks of which the miner is aware. 
\item If $M_P(B, c, t) \neq \bot$ for any coin $c$ that the miner owns, the miner announces some valid block $C$ with $\predecessor(C) = B$.
\item The miner may also announce other blocks. 
\end{itemize}
\end{definition}

Observe that both the honest protocol and Undetectable Nothing-at-Stake are More-Than-Honest. We now present a formal analysis of the mining advantage enjoyed by an attacker using the Undetectable Nothing-at-Stake strategy.

\begin{proposition}\label{prop:UNaS}
Let $P$ be a $D$-Recent Proof-of-Stake protocol where $M_p$ acts as a random oracle with recency $\ell < D$. Let also there be at least $\lambda$ coins in the network, all using a More-Than-Honest strategy. Finally, let $S(B) = \max\{\ell, \predecessor^\ell(B) \neq \bot\}$. Then the Undetectable Nothing-at-Stake strategy announces blocks at least $2-2D/(\lambda+1)$ times the rate as the honest strategy.
\end{proposition}

To help parse the above proposition, it is suggesting that any $D$-Recent protocol which expects at least $\lambda$ coins to be actively mining must defend against Undetectable Nothing-at-Stake whenever $D < \lambda/2$.

\begin{proof}[Proof of Proposition~\ref{prop:UNaS}]
Recall the following notation from the strategy definition: $B$ refers to the block maximizing $S(B)$ over all blocks the attacker is aware of. $B'$ refers to the block maximizing $S(B')$ over all blocks \emph{not} descended from $\predecessor^D(B)$. $C'$ refers to the block built on top of $B'$ that may or may not be announced, pending the safety check.

The only way that the safety check will stop the deviating miner from announcing $C'$ is if the same coin was previously used to mine a block that is a descendant of $\predecessor^D(B)$. This is because $S(B') \geq S(B'')$ for all other $B''$, so no other $B''$ can possibly contribute to a provable deviation. So we just need to analyze the probability that the same coin was previously used to mine a block on top of $\predecessor^D(B)$. 

In order for coin $c$ to have mined a block on top of $\predecessor^D(B)$, it must be the case that the longest chain has not grown by more than $D$ since the last time that $c$ has mined a block (this is a necessary, but not sufficient condition). Taking the converse, this means that a sufficient condition for the safety check to allow the announcement of the illegitimate block is if the longest chain has grown by more than $D$ since the last time that coin $c$ was used to mine a block. So we now just need to understand the fraction of timesteps for which the last time that a block mined by coin $c$ has height within $D$ of the current longest chain, and those for which it doesn't. 

So now, in every timestep, consider the three possible events:
\begin{itemize}
\item Some coin besides $c$ is eligible to mine on top of the longest chain. Because all coins are using a More-Than-Honest strategy, the length of the longest chain grows by one.
\item Coin $c$ is eligible to create a block (regardless of whether or not it chooses to announce it).
\item The longest chain does not grow, and $c$ is not eligible to create a block (but perhaps other miners announce blocks that don't affect the length of the longest chain). 
\end{itemize}

Note first that there are at least $\lambda - 1$ equally likely outcomes corresponding to the first event since there are $\lambda-1$ coins other than $c$. There are two equally likely outcomes corresponding to the second event since the attacker attempts to use the coin $c$ to mine in two independent locations. Thus during every timestep, the second event occurs at most a $2/(\lambda + 1)$ fraction of the time (this would be tight if each other miner was honest, or otherwise only mining on top of a single longest chain). We can now also see that every timestep in which event two happens ``claims'' $D$ timesteps in which event one happens as the timesteps where the safety check could have conceivably prevented announcing the illegitimate block (the next $D$ such timesteps). For all other timesteps, the safety check would definitely have allowed the illegitimate block to be announced. 

So in the limit, only a $2D/(\lambda+1)$ fraction of the ``event-one'' timesteps are claimed, and in all remaining timesteps the safety check would have allowed publication of both blocks. Therefore, a $1-2D/(\lambda+1)$ fraction of the time, the Undetectable Nothing-at-Stake strategy is considering publication of two blocks, and will announce blocks at a rate of at least $2-2D/(\lambda+1)$ times as often as an honest miner. 
\end{proof}

\section{Existing Proof-of-Stake Protocols}\label{sec:examples}
In this section, we discuss several popular Proof-of-Stake protocols, how they fit into our language, and the extent of their susceptibility to the attacks we discuss. Many of these protocols can get quite involved, with several layers of defenses, but our goal is to focus on the core protocols underneath.

\paragraph{Snow White.}
The Snow White protocol~\cite{bentov2016snow} separates time into epochs; each epoch has $T_{epoch}$ time steps. Within each epoch a committee and a hash function/random oracle are decided by looking at blocks in the common history. At each time step, if for some member of the committee with public key $pk$ it holds that $H(pk , time) < Target$, then that member becomes the leader and gets to make a new block. Therefore, the protocol is predictable by our definitions, and vulnerable to predictable selfish mining and predictable double spend.

To address this, Snow White adopts the reward scheme of Fruitchain~\cite{pass2016fruitchains}. That is, blocks don't directly contain transactions, but rather \emph{fruit}, and fruit directly contain transactions. Miners are then rewarded for creating \emph{fruit}, rather than blocks. They further prove the following: if a miner controls an $\alpha < 1/2$ fraction of the stake, and all other miners follow the intended protocol, that miner receives at most a $\alpha+\varepsilon$ fraction of the total rewards (for a small $\varepsilon > 0$ decided by the designer). Their proof has the flavor of a differential privacy guarantee: essentially any strategy that miner uses will in fact result in a $(\alpha-\varepsilon, \alpha+\varepsilon)$ fraction of the total rewards. 

As such, Snow White further notes that it would be more desireable to \emph{strictly} disincentivize known attacks. To this end, they show that the (detectable, because Snow White is not Recent) Nothing-at-Stake attack is strictly unprofitable against their reward scheme. Our work is essentially proposing that it would be worthwhile to do a similar analysis for predictable selfish mining and predictable double-spending. So in summary: (1) Snow White is not vulnerable to undetectable Nothing-at-Stake, because it is not Recent. (2) Snow White is vulnerable to predictable selfish mining, but provably it cannot improve a miner's rewards by more than an additive $\varepsilon$ fraction of the total rewards. (3) Snow White is vulnerable to predictable double spending, which could be mitigated by sufficiently long confirmation times but is otherwise not addressed.

\paragraph{Ouroboros.} At the level of granularity relevant to this paper, Ouroboros is similar to Snow White. The Ouroboros protocol\cite{kiayias2017ouroboros} divides time into epochs, each of which is made up of $L$ time slots. Each of the time slots is assigned to a miner chosen randomly with probability proportional to her stake. Each miner is then eligible to create a block in her assigned time-slot. The set of miners assigned to time slots in an epoch is called the committee for that epoch. The miners in the committee additionally run a secure multiparty coin-flipping protocol in order to generate the randomness needed to select the committee for the next epoch.

The protocol adds an additional category of miner called an ``input-endorser.'' Multiple input-endorsers are randomly selected to each time slot with probability proportional to their stake (using the same source of randomness as the original miners). An input-endorser is responsible for signing the valid transactions they hear about in their given time slot. Miners then include endorsed inputs into their blocks (so endorsed inputs play a similar role to fruit). 

Since miners for the $L$ time-slots in an epoch are determined in advance, the protocol is $L$-globally predictable. The input-endorses introduce additional complexity, but it is also predictable when a miner will become eligible to be an input endorser. The authors of Ouroboros seem aware of the predictable double-spend attack, and require long block confirmation times. For example, suppose that $40\%$ of the stake is controlled by an attacker attempting a double-spend attack. The authors of \cite{kiayias2017ouroboros} compute that in order to achieve $99.9\%$ confidence that this double spend attack cannot succeed, block confirmation times must be at least 148 minutes. If the attacker controls $45\%$ of the stake, confirmation times must be further increased to 663 minutes (these mirror our own calculations).

Like Snow White, Ouroboros proves that as long as a miner controls an $\alpha < 1/2$ fraction of the stake, and all other miners follow the intended protocol, the miner receives a $(\alpha-\varepsilon, \alpha+\varepsilon)$ fraction of the total rewards for essentially any strategy (including honesty, or predictable selfish mining). Similarly to Snow White, our work suggests that it is worthwhile to understand whether predictable selfish mining is indeed profitable. Since the protocol is predictable, Ouroboros is not vulnerable to Undetectable Nothing-at-Stake. 

\paragraph{No Rewards.} Consider either of the aforementioned protocols, but in absence of rewards (i.e. copy Snow White, but don't reward miners for either fruit or blocks). This reward scheme achieves the same formal guarantees as the previous protocols, with $\varepsilon = 0$: if all other miners are following the intended protocol, another miner gains nothing by deviating. Therefore, no-rewards is at least as robust to predictable selfish mining as proved in these prior works. However, no-rewards is also at least as vulnerable to predictable double-spend (even if the confirmation time exceeds the predictability, there is nothing lost by giving it a shot anyway). 

It's not clear whether no-rewards is actually a viable reward scheme in Proof-of-Stake proposals. On one hand, it achieves the same formal guarantees as prior works (in fact, stronger as one can take $\varepsilon = 0$), and one could informally assert that those with stake in the currency have incentive to remain online and follow the protocol to maintain its value. On the other hand, one might equally reasonably worry that without strict incentive to follow the protocol, attempts to double-spend may run rampant. 

By most existing formal measures, no-rewards is at least as incentive compatible as the previously discussed reward schemes are proven to be. The lone exception is that Nothing-at-Stake is provably strictly unprofitable against Snow White, but not against no-rewards. Our work proposes that predictable selfish mining, predictable double-spend, and undetectable nothing-at-stake be given the same treatment in future analyses.

\paragraph{Algorand.} Algorand \cite{micali2016algorand,gilad2017algorand} doesn't fit into our framework, because it is not a Longest-Chain Variant. Algorand is instead based on a Byzantine Consensus protocol. Ideas related to predictable selfish mining are relevant, but there is no formal connection between our work and Algorand.

\paragraph{Bentov-Gabizon-Mizrahi.}

The authors of \cite{bentov2016cryptocurrencies}, design a protocol that selects miners with probability proportional to their stake using a procedure called ``follow-the-Satoshi.'' In this process, a random minimal denomination (one Satoshi) of the currency is chosen and whomever currently owns it is selected as eligible to mine. The authors initially discuss many of the problems with being predictable, and with nothing-at-stake.

In the authors' first proposed protocol, Chains-of-Activity (CoA), the chain is divided into groups of $l$ consecutive blocks. Every miner includes a supposedly random bit in her block, and the concatenation of all the random bits from the $i$-th group of $l$ consecutive blocks is used as a random seed $s_i$. These seeds are then used to run ``follow-the-satoshi'' in order to select miners for later blocks in an interleaved way. In particular seed $s_i$ is used to determine miners for blocks in the $i+2$-th group of $l$ blocks. Thus the CoA protocol is $l$-globally predictable. The authors discuss how transactions should not be considered confirmed until $l$ timesteps, and show that its unlikely for an attacker to succeed in bribing others for a double spend that lasts longer than $l$. However, predictable selfish mining is not discussed.

The authors' follow-up proposal ``Dense-CoA'', randomly chooses $l$ users to be involved in the creation of the $i$-th block. In particular, one special user is selected to actually create the block, but all $l$ of them run a secure coin-flipping protocol to produce the random seed for the next block. Thus, Dense-CoA is $1$-Recent. 

\paragraph{DFinity.} DFinity~\cite{HankeMW17} doesn't cleanly fit inside out framework. In fact, DFinity deviates from our framework as early as the definition of Block. In DFinity, some blocks require a \emph{threshold signature}\footnote{A $k$-of-$n$ threshold signature scheme is such that $n$ participants each have a public and private key, and there is a public poly-time algorithm \textsc{Verify}. For a given message $m$, it is possible for each participant $i$ to generate (in poly-time) a message $S_i(m)$ such that: $\textsc{Verify}(S_1,\ldots, S_n, m) = 1$ if and only if at least $k$ of $n$ inputs are of the form $S_i(m)$. Moreover, one cannot generate a message of the form $S_i(m)$ without the private key for $i$.} from \emph{multiple different miners} to act as a future source of pseudorandomness. The authors are not commenting on the security of DFinity, but note that these ideas aren't widely used (DFinity is the only instance the authors are aware of) as threshold signatures require a trusted setup every time the set of stake-holders eligible to sign blocks changes. Further, a whole new consensus problem arises when deciding (for example) how long to wait for the signatures in the threshold scheme.

\section{Proof-of-Stake and GHOST}\label{app:GHOST}
The Greedy Heaviest-Observed Subtree (GHOST) protocol is an alternative to longest-chain variants, originally proposed for proof-of-work blockchains in \cite{SompolinskyZ13}. The authors observed that even if some blocks did not end up in the main chain, they could still be used as proof that another block should be included. In particular, suppose every miner maintains a tree of blocks that have been published so far. For a given block $B$, every block in the subtree descending from $B$ was created by a miner who believed that block $B$ should be included in the main chain. Thus, one can think of every block in the subtree descending from $B$ as a vote to include $B$. The GHOST protocol is based on this idea. 

The protocol proceeds starting from the root of the tree of blocks and greedily moving down the tree along the fork corresponding to the largest subtree.

\begin{definition}[GHOST Protocol]
For a block $B$ let $W(B)$ denote the size of the subtree rooted at $B$. 
\begin{itemize}
\item Initialize $B$ to be the root of the tree of blocks and do the following:
\begin{enumerate}
\item If $B$ has no direct descendants return $B$
\item Let $C$ be the block maximizing $W(C)$ among all blocks with $\predecessor(C) = B$ in the tree. 
\item Set $B \gets C$ and go to step 1.
\end{enumerate}
\item Attempt to mine on top of the block returned by the above algorithm.
\end{itemize}
\end{definition}

This has advantages in proof-of-work protocols as it allows faster creation of blocks since forks still contribute to the main chain. However, it has serious shortcomings for a $D$-Recent proof-of-stake protocol. The basic reason for this is that even with a small fraction of the total stake, an attacker can greatly increase the size of a chosen subtree by simultaneously mining on top of every node in the subtree. 

For simplicity we will consider a $1$-Recent protocol $P$ which acts as a random oracle with recency $1$. Below we describe an exponential forking attack against GHOST protocols, which allows an attacker with a small fraction of the total stake to force all miners following the GHOST protocol into mining on top of the attacker's chosen subtree.

\begin{definition}[Exponential Forking]
Suppose you wish to introduce a fork rooted at a block $B$. Do the following at every timestep t:
\begin{itemize}
\item For every block $B'$ in the subtree rooted at $B$, and every coin $c$ you own, let $C = M_p(B',c,t)$.
\item If $C \neq \bot$, then announce C.
\end{itemize}
\end{definition}

Now suppose an attacker with an $\alpha$ fraction of the total stake uses the exponential forking attack rooted at a block $B$. Note that by definition, honest miners following the GHOST protocol only attempt to add a block at one location per timestep. Thus, all honest miners will add blocks at a rate of $(1-\alpha)$ per unit time. However, if we let $x_t$ be the size of the subtree rooted at $B$ at time $t$, the attacker will add, on average, $\alpha x_t$ blocks at time $t$. So if we let $y_t$ denote the total number of blocks added by honest miners up until time $t$ we have, on average:
\begin{align*}
x_t &= x_{t-1} + \alpha x_{t-1}\\
y_t &= y_{t-1} + (1 - \alpha).
\end{align*}
Solving both recurrences yields that after $k$ steps we have:
\begin{align*}
x_T &= (1+\alpha)^k x_0\\
y_T &= (1-\alpha)k + y_0.
\end{align*}
That is, the size of the subtree rooted at $B$ grows exponentially faster than the total number of other blocks mined. Thus, all miners following the GHOST Protocol will quickly be forced to mine on the subtree rooted at $B$ produced by the attacker.

\end{document}